\documentclass[copyright,creativecommons]{eptcs}
\providecommand{\event}{ICE 2010} 
\usepackage{breakurl}             
\usepackage{latexsym,amsmath,amssymb}
\usepackage{amsthm}
\usepackage{stmaryrd}
\usepackage{extarrows}

\theoremstyle{definition} 
\newtheorem{definition}{Definition}
\newtheorem{example}{Example}
\newtheorem{theorem}{Theorem}
\newtheorem{lemma}{Lemma}

\DeclareMathOperator{\gram}{::=}
\newcommand{\arros}[1]{\xrightarrow[]{#1}_d}
\newcommand{\arrob}[1]{\xrightarrow[]{#1}_s}
\newcommand{\Arro}[1]{\xLongrightarrow[]{#1}}
\newcommand{\rn}[1]{({\sc #1})}

\def \mathaxiom #1#2{\begin{array}{l}%
    {\mbox{\scriptsize ({\sc #1})} }%
    \\ \iaxiom{#2}%
    \end{array}}
\newcommand{\iaxiom}[1]{\textstyle\rule[-1.3ex]{0cm}{3ex}#1}
\def \mathrule #1#2#3{\begin{array}{l}
    {\mbox{\scriptsize ({\sc #1})} }
    \\ \bigfract{#2}{#3}
\end{array}}
\newcommand{\bigfract}[2]{\frac{^{\textstyle #1}}{_{\textstyle #2}}}

\newcommand{\saga}[1]{\{\![ #1 ]\!\}}
\newcommand{\sagarun}[2]{\{\![ #1 , #2 ]\!\}}
\newcommand{\prot}[1]{\llbracket #1 \rrbracket_{\commit}}
\newcommand{\killed}[1]{\llbracket #1 \rrbracket_{\abort}}
\newcommand{\protany}[1]{\llbracket #1 \rrbracket_{\dontcare}}
\newcommand{\extr}[1]{extr( #1 )}
\newcommand{\nullp}[1]{null( #1 )}
\newcommand{\abort}{\boxtimes}
\newcommand{\commit}{\ensuremath{\boxdot}}
\newcommand{\dontcare}{\square}
\newcommand{\crash}{\boxast}

\newcommand{\fabort}{\overline{\abort}}
\newcommand{\fcrash}{\overline{\crash}}
\newcommand{\fabcrash}{\overline{\boxplus}}
\newcommand{\nil}{\mathbf{0}}

\DeclareMathOperator{\lin}{lin}

\title{Static vs Dynamic SAGAs\thanks{Research partially supported by the Project FP7-231620 {\sc HATS}.}}
\author{Ivan Lanese
\institute{Focus Team, University of Bologna/INRIA\\ Bologna, Italy}
\email{lanese@cs.unibo.it}}
\def\titlerunning{Static vs dynamic SAGAs}
\def\authorrunning{I. Lanese}
\begin{document}
\maketitle

\begin{abstract}
SAGAs calculi (or simply SAGAs) have been proposed by Bruni et al. as
a model for long-running transactions. The approach therein can be
considered static, while a dynamic approach has been proposed by
Lanese and Zavattaro. In this paper we first extend both static SAGAs
(in the centralized interruption policy) and dynamic SAGAs to deal
with nesting, then we compare the two approaches.
\end{abstract}

\section{Introduction}
Computing systems are becoming more and more complex, composed by a
huge number of components interacting in different ways. Also,
interactions are frequently loosely-coupled, in the sense that each
component has scarce information on its communication partners, which may
be unreliable (e.g., they may disconnect, or may not follow the expected
protocol). Communication may be unreliable too, for
instance in the case of wireless networks. Nevertheless, applications
are expected to provide reliable services to their users. For these
reasons, a main concern is the management of unexpected events.

In the case of loosely-coupled distributed systems (e.g., for web
services), unexpected events are managed according to the
\emph{long-running transaction} approach. A long-running transaction
is a computation that either \emph{commits} (i.e., succeeds), or it
\emph{aborts} and is \emph{compensated}. Compensating a (long-running)
transaction means executing a sequence of actions that revert the
effect of the actions that lead to abortion, so as to reach a
consistent state. This is a relaxation of the properties of ACID
transactions from database theory~\cite{HR:ACID}, based on the fact
that in the systems we are interested in rollback cannot always be
perfect (e.g., one can not undo the sending of an e-mail, and if one
tries to undo an airplane reservation, (s)he may have to pay some
penalty).

Recently, many proposals of formal models to reason about properties
of long-running transactions, and about systems exploiting them, have
been put forward. We concentrate on process calculi, since they are a
good tool to experiment with different primitives and compare their
relative merits and drawbacks. Later on, the results of these
experiments can drive the design of real languages. Process calculi
approaches to long-running transactions divide in two main categories:
\emph{interaction-based calculi} and \emph{flow composition
  approaches}.  Interaction-based calculi are obtained by extending
name passing calculi with dedicated primitives, and one of their main
concerns is the interplay between communication and transactions. We
recall among them the $\pi$t-calculus~\cite{BLZ:PT},
c-join~\cite{BMM:CJOIN}, web$\pi$~\cite{LZ:WEBPI},
dc$\pi$~\cite{VFR:DCP}, the ATc calculus~\cite{BT:ATc} and
SOCK~\cite{ACSD}.  Flow composition approaches instead deal with the
composition of atomic activities, studying how to derive compensations
for complex activities from compensations of basic ones.  We recall
for instance SAGAs~\cite{GGKKS:SAGAS}, StAC~\cite{BF:STAC},
cCSP~\cite{BHF:CCSP} and the SAGAs calculi~\cite{BMM:SAGASPOPL}.  Some
of the primitives for long-running transactions have been introduced
in real languages such as WS-BPEL~\cite{WSBPEL} and
Jolie~\cite{ECOWS:JOLIEFAULT}. Long-running transactions have also
been analyzed in a choreographic setting in~\cite{CHY:EXC}. However,
only a few works~\cite{BBF+:SAGASCONCUR,SEFM09,ESOP2010} until now
have tried to clarify the relationships between the different
approaches. We want to go forward in the understanding of those
relationships.

As shown in~\cite{ESOP2010}, a main distinction to be done is between
\emph{static compensations}, where essentially the order of execution
of compensations depends on the syntactic structure of the program,
and \emph{dynamic compensations}, where it depends on the order of
execution of the activities at runtime and cannot be determined
statically. The analysis in~\cite{ESOP2010} has been carried on in the
case of interaction-based calculi. However compensations are also
heavily studied in the framework of flow composition languages. The
only dynamic approach to compensations in the framework of flow
composition languages we are aware of is the one of dynamic SAGAs
\cite{SEFM09}. There however only non-nested SAGAs have been defined,
and they have been contrasted with the dynamic interaction-based
calculus SOCK~\cite{ACSD}, but not with the classic static SAGAs
calculi~\cite{BMM:SAGASPOPL}. Here we want to carry on this last
comparison. More precisely, since different flavors of static SAGAs
calculi exist, we contrast them with the {\em centralized
  interruption} policy defined in~\cite{BBF+:SAGASCONCUR}. Here
``interruption'' means that, if there are many concurrent flows of
computation and one of them aborts, the other ones are stopped.
``Centralized'' instead means that the compensation procedure of
parallel flows of computation is started only after all forward flows
have been stopped.  We have chosen this approach since it is the one
that more closely matches the behavior of systems programmed in
WS-BPEL or Jolie. Actually, also this flavor of SAGAs has been defined
only in the non-nested case, while we are interested in the nested
case too. In fact, nested SAGAs are fundamental to model complex flows
of activities.  Thus the contributions of this paper are:
\begin{itemize}
\item a description of the semantics of nested SAGAs, both under the
  static centralized interruption approach (Section~\ref{sec:static})
  and the dynamic approach (Section~\ref{sec:dynamic}); both the
  extensions are non trivial, as we will show when we present their
  semantics;
\item a comparison between the two semantics (Section~\ref{sec:comp}),
  showing that the computations allowed by the dynamic semantics are a
  strict subset of the ones allowed by the static semantics, and this
  is due to the fact that the dynamic semantics is more strict on the
  possible orders of execution of compensations of parallel
  activities; this comparison has also been used as a sanity check for
  the formalization of the two semantics.
\end{itemize}

\section{Static SAGAs}\label{sec:static}
SAGAs calculi~\cite{BMM:SAGASPOPL} (SAGAs from now on) are calculi for
defining compensating processes. A process is composed of
\emph{activities}, ranged over by $A,B,\dots$, and each activity may
have its own compensating activity. Processes form long-running
transactions, called \emph{sagas} in this context. A saga either
\emph{commits} (i.e., succeeds), or it \emph{aborts} and is
\emph{compensated}.  Abortion of the compensation causes a
\emph{failure}, which is recognized as a catastrophic event and
terminates the whole process. We are interested in nested sagas, thus
sagas are processes too.

\begin{definition}[Sagas]
Saga processes are defined by the following grammar:\\
\noindent
$P \gram \nil \ | \ A \ | \ A \div B \ | \ P;P \ | \ P|P \ | \ \saga{P}$
\end{definition}

Processes can be the empty activity $\nil$, an activity $A$ without a
specified compensation, or an activity $A \div B$ specifying $B$ as
compensation for $A$. Processes can be composed in sequence ($P;P$) or
in parallel ($P|P$). Sagas can be nested, thus a saga $\saga{P}$ is a
process too.  In the following we will disregard activities $A$, since
they can be considered as a particular instance of compensable activities $A
\div B$ where $B=0$.

\begin{table*}[p]
\[
\begin{array}{l@{\hspace{.1cm}}l}
  \mathaxiom{zero}
	    {\Gamma \vdash \langle 0, \beta \rangle \arrob{0} \langle \commit, \beta \rangle} 
& \mathaxiom{s-act} 
	    {A\mapsto\commit,\Gamma\vdash
	      \langle A\div B, \beta \rangle \arrob{A} \langle \commit, {B;\beta} \rangle} 
\\[15pt]
  \mathaxiom{f-act} 
	    {A\mapsto\abort,\Gamma\vdash
	      \langle A\div B, \beta \rangle \arrob{0} \langle \abort, {\beta} \rangle} 
&
  \mathrule{s-step}
	   {\Gamma \vdash \langle P, \beta \rangle \arrob{\alpha} \langle \commit, \beta'' \rangle 
	    \quad
	    \Gamma \vdash \langle Q, \beta'' \rangle \arrob{\alpha'} \langle \dontcare, \beta' \rangle} 
	   {\Gamma \vdash \langle {P;Q}, \beta \rangle \arrob{\alpha;\alpha'} \langle \dontcare, \beta' \rangle} 
\\[15pt]
  \mathrule{a-step} 
	   {\Gamma \vdash \langle P, \beta \rangle \arrob{\alpha} \langle \dontcare, \beta' \rangle \quad \dontcare \neq \commit} 
	   {\Gamma \vdash \langle {P;Q}, \beta \rangle \arrob{\alpha} \langle \dontcare, \beta' \rangle}
\\[15pt]
\multicolumn{2}{l}{
  \mathrule{s-par}
	   {\Gamma \vdash \langle P, 0 \rangle \arrob{\alpha_1} \langle \dontcare_1, \beta_1 \rangle 
	    \quad
	    \Gamma \vdash \langle Q, 0 \rangle \arrob{\alpha_2} \langle \dontcare_2, \beta_2 \rangle \quad \dontcare_1,\dontcare_2 \in \{\commit,\abort,\fabort\}} 
	   {\Gamma \vdash \langle {P|Q}, \beta \rangle \arrob{\alpha_1|\alpha_2} \langle \dontcare_1 \land \dontcare_2, (\beta_1|\beta_2);\beta \rangle}	   
}
\\
\multicolumn{2}{l}{
  \mathrule{f-par}
	   {\Gamma \vdash \langle P, 0 \rangle \arrob{\alpha_1} \langle \dontcare_1, \beta_1 \rangle 
	    \quad
	    \Gamma \vdash \langle Q, 0 \rangle \arrob{\alpha_2} \langle \dontcare_2, \beta_2 \rangle \quad \dontcare_2 \in \{\crash,\fcrash,\fabcrash\}} 
	   {\Gamma \vdash \langle {P|Q}, \beta \rangle \arrob{\alpha_1|\alpha_2} \langle \dontcare_1 \land \dontcare_2, 0\rangle}	   
}
\\[15pt]
  \mathaxiom{forced-abt}
            {\Gamma \vdash \langle P, \beta \rangle \arrob{0} \langle \fabort, \beta \rangle}
&
  \mathaxiom{forced-fail}
	    {\Gamma \vdash \langle P, \beta \rangle \arrob{0} \langle \fcrash, 0 \rangle}
\\[15pt] 
 \mathrule{sub-cmt} 
	   {\Gamma\vdash \langle P, 0 \rangle \arrob{\alpha} \langle \commit, \beta' \rangle} 
	   {\Gamma\vdash
	     \langle \saga{P}, \beta\rangle \arrob{\alpha} \langle \commit,{\beta';\beta} \rangle} 
&
  \mathrule{sub-abt}
     { \Gamma \vdash \langle P, 0 \rangle \arrob{\alpha} \langle \abort, \beta' \rangle \quad \Gamma \vdash \langle \beta', 0 \rangle \arrob {\beta'} \langle \commit, 0 \rangle}
	   {\Gamma\vdash
	     \langle \saga{P}, \beta \rangle \arrob{\alpha;\beta'} \langle \commit, \beta \rangle} 
\\[15pt]  
  \mathrule{sub-fail-1} 
     { \Gamma \vdash \langle P, 0 \rangle \arrob{\alpha} \langle \dontcare, 0 \rangle \quad \dontcare \in \{\crash, \fcrash, \fabcrash\}}
	   {\Gamma\vdash
	     \langle \saga{P}, \beta \rangle \arrob{\alpha} \langle \dontcare , 0 \rangle} 
&
  \mathrule{sub-fail-2} 
     { \Gamma \vdash \langle P, 0 \rangle \arrob{\alpha} \langle \abort, \beta' \rangle \quad \Gamma \vdash \langle \beta', 0 \rangle \arrob{\beta''} \langle \abort, 0 \rangle}
	   {\Gamma\vdash
	     \langle \saga{P}, \beta \rangle \arrob{\alpha;\beta''} \langle \crash , 0 \rangle} 
\\[15pt]
\multicolumn{2}{l}{
\mathrule{sub-forced-1} 
     { \Gamma \vdash \langle P, 0 \rangle \arrob{\alpha} \langle \abort, \beta' \rangle \quad \Gamma \vdash \langle \beta', 0 \rangle \arrob{\beta''} \langle \fcrash, 0 \rangle}
	   {\Gamma\vdash
	     \langle \saga{P}, \beta \rangle \arrob{\alpha;\beta''} \langle \fcrash , 0 \rangle} 
}
\\[15pt]
\multicolumn{2}{l}{
\mathrule{sub-forced-2}
	   {\Gamma\vdash \langle P, 0 \rangle \arrob{\alpha} \langle \fabort, \beta' \rangle
	     \quad
	     \Gamma\vdash \langle \beta', 0 \rangle \arrob{\beta''} \langle \dontcare_1, 0 \rangle
	   }
	   {\Gamma\vdash
	     \langle \saga{P}, \beta \rangle \arrob{\alpha;\beta''} \langle \dontcare_2, 0 \rangle} 
       \dontcare_2 = \Big\{
         \begin{array}{ll}
	   \fabort & \mbox{if } \dontcare_1 = \commit \\
	   \fabcrash & \mbox{if } \dontcare_1 \in \{\abort,\crash\}
	 \end{array}
}
\end{array}
\]
\caption{Static semantics of nested SAGAs.}
\protect\label{table:staticsem}
\end{table*}

The idea underlying the static SAGA semantics is that compensations of
sequential activities are executed in reverse order, while
compensations of parallel activities are executed in parallel. In
particular, the possible orders of execution for compensation
activities can be determined statically, looking at the structure of
the process. However, as shown by~\cite{BBF+:SAGASCONCUR}, different
design choices concerning the behavior of parallel activities are
possible. As already said, we consider the semantics with
\emph{interruption of parallel activities and centralized
  compensations} proposed in~\cite{BBF+:SAGASCONCUR}. According to
this semantics parallel activities are stopped when one of them
aborts, while in the semantics without interruption they are run to
the end (and then compensated). Also, compensations are handled in a
centralized way (but for subtransactions), while in the semantics with
distributed compensations each flow is responsible for executing its
own compensations. The semantics presented in~\cite{BBF+:SAGASCONCUR}
however does not consider nested sagas, while we consider them
important both from a theoretical and a practical point of view. From
a practical point of view, nesting is fundamental to model large
systems in a compositional way. From a theoretical point of view
nesting raises interesting questions on the interplay between the
behavior of a saga and of its subsagas. For this reason we extend the
semantics to deal with nested SAGAs, taking inspiration from the one
in \cite{BMM:SAGASPOPL}, which has however distributed compensations.

\begin{definition}[Static semantics of SAGAs]
The static semantics $\arrob{}$ of SAGAs is the LTS defined in
Table~\ref{table:staticsem} (we assume symmetric
rules for parallel composition).
\end{definition}

A saga may commit, abort or fail, denoted respectively by $\commit$,
$\abort$ and $\crash$. Also, a saga may acknowledge an external
abortion or failure, and these two possibilities are denoted by
$\fabort$ and $\fcrash$ respectively.  Finally, a saga may answer an
external abortion with a failure (when an external abort causes it to
compensate a subsaga, and the compensation fails), denoted as
$\fabcrash$. Note that this situation never occurs without nesting. In
fact, under the centralized compensation approach, only subsagas are
compensated locally, while other processes are compensated at the top
level.

The behavior of a saga is determined by the behavior of its
constituent activities, which is specified by an environment $\Gamma$
mapping each activity to either $\commit$ or $\abort$. The semantics
of SAGAs is given as a relation $\Gamma \vdash \langle P, \beta
\rangle \arrob{\alpha} \langle \dontcare, \beta' \rangle$, defined in
the big-step style. Here label $\alpha$ is the observation, showing the
successfully executed activities. Observations are obtained by
combining activities in sequence and in parallel. If one considers an
interleaving setting, label $A|B$ can be considered as a shortcut for the
two possible sequences, $A;B$ and $B;A$.  We consider observations up
to the following axioms: $0;\alpha=\alpha$, $\alpha;0=\alpha$,
$0|\alpha=\alpha$, $\alpha|0=\alpha$.  Also, $\beta$ is the
compensation stored for execution at the beginning of the computation
and $\beta'$ the final stored compensation. Finally, $\dontcare$
ranges over $\{\commit, \abort, \crash, \fabort, \fcrash,
\fabcrash\}$, the possible outcomes of the saga.

The first three rules execute the empty activity and basic
activities. Note that rule \rn{f-act} does not execute the
compensation (differently from the rules for distributed compensations
in the literature), since this will be executed in a centralized way
(see rule \rn{sub-abt}).  Rule \rn{s-step} deals with sequential
composition when the first part of the computation succeeds. Rule
\rn{a-step} deals with all the other cases. Rules \rn{s-par} and
\rn{f-par} concern parallel composition. The operator $\land$ in these
rules is the symmetric closure of the one defined in the table below:
\[
\begin{array}{l@{\hspace{.5cm}}|@{\hspace{.5cm}}l@{\hspace{1cm}}l@{\hspace{1cm}}l@{\hspace{1cm}}l@{\hspace{1cm}}l@{\hspace{1cm}}l}
\quad \land       & \commit   & \abort & \fabort &\crash & \fcrash & \fabcrash\\ 
\hline
\quad\commit   & \commit   \\
\quad\abort    & \abort    & -      & \\
\quad\fabort   & \fabort   & \abort & \fabort & \\
\quad\crash    & \crash    & -      & -       & -      & \\
\quad\fcrash   & \fcrash   & -      & -       & \crash & \fcrash   \\
\quad\fabcrash & \fabcrash & \crash & -       & -      & \fabcrash & - \\
\end{array}
\]
The two rules differ since in the first case the compensation is
stored, in the second one it is discarded (since a failure is
involved). Rules \rn{forced-abt} and \rn{forced-fail} show that a
process can be stopped either by an external abort or by an external
failure. In the second case the compensation is discarded. Rule
\rn{sub-cmt} allows a saga to commit upon commit of its internal
computation. Rule \rn{sub-abt} instead allows a saga to commit upon
abort of the internal computation and successful compensation. Rule
\rn{sub-fail-1} propagates to a saga a catastrophic outcome of its
internal computation. Rule \rn{sub-fail-2} establishes failure for an
aborted saga whose compensation aborts too. Rule \rn{sub-forced-1}
allows an external failure to interrupt a compensating saga. Finally,
rule \rn{sub-forced-2} deals with external requests of abortion for
sagas. The saga is interrupted and compensated. If the compensation is
not successful then a failure is propagated to the outer level. Note
that it is not possible to make a saga abort while it is executing its
own compensation: the execution of a compensation is protected from
further aborts. Compensations of sagas are executed locally, in
parallel with the normal flow of external activities, and before
starting the compensations of external sagas containing them.

We show now a few derivable transitions to clarify the semantics.

\begin{example}\label{ex:shipstatic}
Consider a ship for transporting goods. Assume that two different
kinds of goods, $A$ and $B$, have to be loaded, and the order is not
relevant. Also, $A$ is not necessary, while $B$ is. After loading the
ship can leave. This can be modeled using a SAGA process $P=(\saga{loadA
\div unloadA} | loadB \div unloadB);leave$. Assume that all the
activities but $leave$ succeed.  We can derive as a big-step:\\
\noindent
$\Gamma \vdash \langle P,0 \rangle \arrob{loadA|loadB} \langle \abort, unloadA|unloadB \rangle$\\
thus the process has aborted.
If we put the whole process inside a saga then compensation is actually executed and the saga succeeds:\\
\noindent
$\Gamma \vdash \langle \saga{P},0 \rangle \arrob{(loadA|loadB);(unloadA|unloadB)} \langle \commit, 0 \rangle$

Assume now that all the activities but $loadB$ and $unloadA$
succeed. If the failure of $loadB$ occurs before the execution of
$loadA$ then we have the transition:\\
\noindent
$\Gamma \vdash \langle P,0 \rangle \arrob{0} \langle \abort, 0 \rangle$\\
derived using rule \rn{s-par} where the left component performs
$\fabort$ (derived using rule \rn{sub-forced-2} with left premise
$\Gamma \vdash \langle loadA \div unloadA, 0 \rangle \arrob{0} \langle \fabort, 0
\rangle$) and the right one performs $\abort$.

If the failure of $loadB$ occurs after the execution of
$loadA$ we have instead the transition:\\
\noindent
$\Gamma \vdash \langle P,0 \rangle \arrob{loadA} \langle \crash, 0 \rangle$\\
derived using rule \rn{s-par} where the left component performs
$\fabcrash$ (derived using rule \rn{sub-forced-2} with left premise
$\Gamma \vdash \langle loadA \div unloadA, 0 \rangle \arrob{loadA}
\langle \fabort, unloadA \rangle$) and the right one performs
$\abort$.

Having $\fabort$ or $\fcrash$ instead of $\fabcrash$ (which is novel
of this semantics) would not faithfully model the intuition. In fact,
in the first case the result of the transition would be $\abort$
instead of $\crash$, while in the second case the transition would not
be derivable at all, since $\abort \land \fcrash$ is undefined
(otherwise an abort could make a transaction fail, even if
compensations were successful).
\end{example}

\begin{example}\label{ex:subtransstatic}
We consider here a modification of the example above, so to clarify
another aspect of the semantics. We consider a SAGA process
$P'=\saga{loadA1 \div unloadA1;loadA2 \div unloadA2} | (loadB1 \div
unloadB1;loadB2 \div unloadB2)$, where each load activity has been
split in two subactivities.  Assume that activity $loadB2$ aborts,
while all the other activities succeed.  On the right-hand side we
have a transition:\\
\noindent 
$\Gamma \vdash \langle loadB1 \div unloadB1; loadB2
\div unloadB2, 0 \rangle \arrob{loadB1} \langle \abort, unloadB1
\rangle$\\ 
This interacts with a left transition of the form:\\ 
\noindent
$\Gamma
\vdash \langle \saga{loadA1 \div unloadA1;loadA2 \div unloadA2},0 \rangle \arrob{loadA1;unloadA1} \langle 0, \fabort
\rangle$\\ 
Thus the label of the whole transition is
$(loadA1;unloadA1)|loadB1$. In particular, the compensation of the
left branch, $unloadA1$, can be executed before $loadB1$, i.e.\ before
the fault actually happens. This can be justified by considering an
asynchronous scenario, where the observer receives events from
different parallel processes out-of-order. The same problem occurs with the
distributed semantics~\cite{BBF+:SAGASCONCUR}, and is due to the fact
that sagas are compensated locally. We will see that the dynamic
semantics solves this problem. An approach for solving the problem
also in the static scenario can be found in~\cite{BKLS:COORDINATED}.
\end{example}


\section{Dynamic SAGAs}\label{sec:dynamic}
Dynamic SAGAs have been proposed in \cite{SEFM09}, in the non-nested
case, to reduce the degree of nondeterminism in saga execution. In
static SAGAs, in fact, the compensation of $A_1\%B_1|A_2\%B_2$ is
$B_1|B_2$. Thus both the orders $B_1;B_2$ and $B_2;B_1$ are allowed,
independently of the order in which activities $A_1$ and $A_2$ have been
executed. Dynamic SAGAs specify instead that compensations of parallel
activities are executed in reverse order: if the execution of normal
activities has been $A_1;A_2$ then the execution of compensations is
$B_2;B_1$. Thus the order of execution of compensations depends on
runtime information on the execution of the basic activities.  While
the semantics of static SAGAs is normally given in the big-step style,
the semantics of dynamic SAGAs is given in the small-step one. A more
detailed motivation for this will be given in the next section.

Static and dynamic SAGAs have the same syntax, differing only in the
semantics. However, to define the semantics of dynamic SAGAs, we find it
convenient to exploit an extended syntax:\\

\noindent
$P \gram \dots \ | \ \sagarun{P}{P} \ | \ \prot{P} \ | \ \killed{P}$\\

Here $\sagarun{P}{\beta}$ is a running saga, where $P$ is the body and
$\beta$ a stored compensation (syntactically, a process obtained as
sequential composition of basic activities). From now on, $\saga{P}$
stands for $\sagarun{P}{0}$. Also, $\prot{P}$ and $\killed{P}$ are
executing compensations. Notation $\protany{P}$ ranges over both of
them. Compensations should be executed in a protected way since we do
not want further abortions to stop them. Similar solutions are
exploited for instance in SOCK~\cite{ACSD}, WS-BPEL~\cite{WSBPEL} and
others. The difference between $\prot{P}$ and $\killed{P}$ is that
$\prot{P}$ has been triggered by the transition itself and commits if
$P$ commits, while $\killed{P}$ has been activated by an external
abort, thus it has to re-raise the abort if $P$ commits.

We need to extract the protected compensations from a process (to
actually protect them, cfr.\ rule \rn{a-par-d}), and we will use to
this end the function $\extr{\bullet}$ defined below. Similar
functions are used, e.g., in SOCK~\cite{ACSD} and dc$\pi$~\cite{VFR:DCP}.

\[
\begin{array}{rcl}
\extr{0} & = & 0\\
\extr{A \div B} & = & 0\\
\extr{P;Q} & = & \extr{P}\\
\extr{P|Q} & = & \extr{P} | \extr{Q}\\
\extr{\sagarun{P}{\beta}} & = & \extr{P};\beta\\
\extr{\protany{P}} & = & P
\end{array}
\]

Finally, we assume a predicate $\nullp{P}$, which holds if $P$ has no behavior:
\[
\begin{array}{rcl}
\nullp{0} & = & true\\
\nullp{P;Q} & = & \nullp{P} \land \nullp{Q}\\
\nullp{P|Q} & = & \nullp{P} \land \nullp{Q}\\
\nullp{\sagarun{P}{\beta}} & = & \nullp{P}\\
\nullp{\protany{P}} & = & \nullp{P}\\
\nullp{\protany{P}} & = & false \textrm{ otherwise}
\end{array}
\]

The dynamic semantics of SAGAs that we present extends the one in
\cite{SEFM09} to deal with nested sagas. The extension is non trivial:
for instance for the non-nested case neither the function
$\extr{\bullet}$ nor the auxiliary runtime syntax were needed.

\begin{definition}[Dynamic semantics of SAGAs]
The dynamic semantics $\arros{}$ of SAGAs is the LTS defined in
Table~\ref{table:dynamicsem} (we assume symmetric
rules for parallel composition). 
\end{definition}

\begin{table*}[p]
{\small
\[
\begin{array}{l@{\hspace{.1cm}}l}
  \mathaxiom{zero-d}
	    {\Gamma \vdash \langle 0, \beta \rangle \arros{0} \langle \commit, \beta \rangle} 
& \mathaxiom{s-act-d} 
	    {A\mapsto\commit,\Gamma\vdash
	      \langle A\div B, \beta \rangle \arros{A} \langle \commit, {B;\beta} \rangle} 
\\[5pt]
  \mathaxiom{f-act-d} 
	    {A\mapsto\abort,\Gamma\vdash
	      \langle A\div B, \beta \rangle \arros{0} \langle \abort, {\beta} \rangle} 
&
  \mathrule{step-d}
	   {\Gamma \vdash \langle P, \beta \rangle \arros{a} \langle P', \beta' \rangle} 
	   {\Gamma \vdash \langle {P;Q}, \beta \rangle \arros{a} \langle P';Q, \beta' \rangle} 
\\
  \mathrule{k-step-d}
	   {\Gamma \vdash \langle P, \beta \rangle \arros{\dagger} \langle P', \beta' \rangle} 
	   {\Gamma \vdash \langle {P;Q}, \beta \rangle \arros{\dagger} \langle P', \beta' \rangle} 
&
  \mathrule{s-step-d}
	   {\Gamma \vdash \langle P, \beta \rangle \arros{a} \langle \commit, \beta' \rangle} 
	   {\Gamma \vdash \langle {P;Q}, \beta \rangle \arros{a} \langle Q, \beta' \rangle} 
\\
  \mathrule{a-step-d} 
	   {\Gamma \vdash \langle P, \beta \rangle \arros{0} \langle \abort, \beta' \rangle} 
	   {\Gamma \vdash \langle {P;Q}, \beta \rangle \arros{0} \langle \abort, \beta' \rangle} 
&
  \mathrule{f-step-d} 
	   {\Gamma \vdash \langle P, \beta \rangle \arros{0} \langle \crash, 0 \rangle} 
	   {\Gamma \vdash \langle {P;Q}, \beta \rangle \arros{0} \langle \crash, 0 \rangle} 
\\
  \mathrule{par-d}
	   {\Gamma \vdash \langle P, \beta \rangle \arros{a} \langle P', \beta' \rangle} 
	   {\Gamma \vdash \langle {P|Q}, \beta \rangle \arros{a} \langle P'|Q, \beta' \rangle}
&
  \mathrule{k-par-d}
	   {\Gamma \vdash \langle P, \beta \rangle \arros{\dagger} \langle P', \beta' \rangle} 
	   {\Gamma \vdash \langle {P|Q}, \beta \rangle \arros{\dagger} \langle P'|\killed{\extr{Q}}, \beta' \rangle}
\\
  \mathrule{s-par-d}
	   {\Gamma \vdash \langle P, \beta \rangle \arros{a} \langle \commit, \beta' \rangle} 
	   {\Gamma \vdash \langle {P|Q}, \beta \rangle \arros{a} \langle Q , \beta' \rangle}	   
&
  \mathrule{a-par-d}
	   {\Gamma \vdash \langle P, \beta \rangle \arros{0} \langle \abort, \beta' \rangle \quad \extr{Q}=\beta'' \quad \neg \nullp{\beta''}} 
	   {\Gamma \vdash \langle {P|Q}, \beta \rangle \arros{\dagger} \langle \killed{\beta''}, \beta' \rangle}	   
\\
  \mathrule{a-par-fin-d}
	   {\Gamma \vdash \langle P, \beta \rangle \arros{0} \langle \abort, \beta' \rangle \quad \extr{Q}=\beta'' \quad \nullp{\beta''}} 
	   {\Gamma \vdash \langle {P|Q}, \beta \rangle \arros{0} \langle \abort, \beta' \rangle}	   
&
  \mathrule{f-par-d}
	   {\Gamma \vdash \langle P, \beta \rangle \arros{0} \langle \crash, 0 \rangle} 
	   {\Gamma \vdash \langle {P|Q}, \beta \rangle \arros{0} \langle \crash, 0 \rangle}	   
\\
  \mathrule{saga-d}
     { \Gamma \vdash \langle P, \beta \rangle \arros{a} \langle P', \beta' \rangle}
     { \Gamma \vdash \langle \sagarun{P}{\beta}, \beta'' \rangle \arros{a} \langle \sagarun{P'}{\beta'},\beta'' \rangle}	 
&
  \mathrule{k-saga-d}
     { \Gamma \vdash \langle P, \beta \rangle \arros{\dagger} \langle P', \beta' \rangle}
     { \Gamma \vdash \langle \sagarun{P}{\beta}, \beta'' \rangle \arros{0} \langle \sagarun{P'}{\beta'},\beta'' \rangle}	 
\\
  \mathrule{s-saga-d}
     { \Gamma \vdash \langle P, \beta \rangle \arros{a} \langle \commit, \beta' \rangle}
     { \Gamma \vdash \langle \sagarun{P}{\beta}, \beta'' \rangle \arros{a} \langle \commit, \beta';\beta'' \rangle}	 
&
  \mathrule{a-saga-d}
     { \Gamma \vdash \langle P, \beta \rangle \arros{0} \langle \abort, \beta' \rangle}
     { \Gamma \vdash \langle \sagarun{P}{\beta}, \beta'' \rangle \arros{0} \langle \prot{\beta'}, \beta'' \rangle}	 
\\
  \mathrule{f-saga-d}
     { \Gamma \vdash \langle P, \beta \rangle \arros{0} \langle \crash, 0 \rangle}
     { \Gamma \vdash \langle \sagarun{P}{\beta}, \beta'' \rangle \arros{0} \langle \crash, 0 \rangle}	 
&
  \mathrule{prot-d}
	   {\Gamma \vdash \langle P, \beta \rangle \arros{a} \langle P', \beta' \rangle} 
	   {\Gamma \vdash \langle \protany{P}, \beta \rangle \arros{a} \langle \protany{P'}, \beta' \rangle} 
\\
  \mathrule{k-prot-d}
	   {\Gamma \vdash \langle P, \beta \rangle \arros{\dagger} \langle P', \beta' \rangle} 
	   {\Gamma \vdash \langle \killed{P}, \beta \rangle \arros{\dagger} \langle \killed{P'}, \beta' \rangle} 
&
  \mathrule{s-prot-d}
	   {\Gamma \vdash \langle P, \beta \rangle \arros{a} \langle \commit, \beta' \rangle} 
	   {\Gamma \vdash \langle \prot{P}, \beta \rangle \arros{a} \langle \commit, \beta' \rangle} 
\\
  \mathrule{s-killed-d}
	   {\Gamma \vdash \langle P, \beta \rangle \arros{a} \langle \commit, \beta' \rangle} 
	   {\Gamma \vdash \langle \killed{P}, \beta \rangle \arros{a} \langle \abort, \beta' \rangle} 
&
  \mathrule{a-prot-d} 
	   {\Gamma \vdash \langle P, \beta \rangle \arros{0} \langle \abort, \beta' \rangle} 
	   {\Gamma \vdash \langle \protany{P}, \beta \rangle \arros{0} \langle \crash, 0 \rangle}
\end{array}
\]
}
\caption{Dynamic semantics of nested SAGAs.}
\protect\label{table:dynamicsem}
\end{table*}

Basic steps are as for the standard semantics.  Rules for composition
operators allow deriving both intermediate steps $\Gamma \vdash
\langle P, \beta \rangle \arros{a} \langle P', \beta' \rangle$ and
final steps $\Gamma \vdash \langle P, \beta \rangle \arros{a} \langle
\dontcare, \beta' \rangle$ (here $\dontcare$ ranges over
$\{\commit,\abort,\crash\}$ and $a$ is an activity name). Also $\dagger$ is
a possible label, denoting an abortion which is delayed to wait for
termination of running compensation activities. Rules \rn{step-d},
\rn{k-step-d}, \rn{s-step-d}, \rn{a-step-d} and \rn{f-step-d} deal
with the possible evolutions of sequential composition. Rules
\rn{par-d}, \rn{s-par-d} and \rn{f-par-d} concern normal computation,
commit and failure of one branch of parallel composition,
respectively. Rule \rn{a-par-d} deals with abortion of one branch. If
the other branch includes some running compensations, then abortion is
delayed and compensation execution is completed first. Running
compensations are extracted by function $\extr{P}$ and thus preserved,
while other parallel activities are discarded. Delayed abortion is
propagated using label $\dagger$. Label $\dagger$ is propagated by
rule \rn{k-par-d}, extracting running compensations from parallel
processes. When all the compensations have been completed (rule
\rn{a-par-fin-d}) abortion is raised again. Rules \rn{saga-d},
\rn{s-saga-d} and \rn{f-saga-d} deal with normal computation, commit
and failure of the internal computation in a saga, respectively. Rule
\rn{k-saga-d} stops the propagation of label $\dagger$. Rule
\rn{a-saga-d} deals with abortion of the internal computation of a
saga: the stored compensation is executed in a protected way.  The
behavior of protection (of the two kinds) is defined by rules
\rn{prot-d} for normal steps, rule \rn{k-prot-d} for delayed abortion
(actually, this can happen only for $\killed{P}$) and rule
\rn{a-prot-d} for abortion (producing a failure). The two kinds of
protection differ in case of commit of the internal computation:
$\prot{P}$ commits (rule \rn{s-prot-d}) while $\killed{P}$ aborts
(rule \rn{s-killed-d}), re-raising the delayed abortion.

We show a few computations as examples.

\begin{example}\label{ex:shipdynamic}
Let us consider the saga process defined in Example \ref{ex:shipstatic}.
Remember that $P=(\saga{loadA
\div unloadA} | loadB \div unloadB);leave$. Assume that all the
activities but $leave$ succeed.  We can derive, e.g., the following computation:
\begin{eqnarray*}
\Gamma \vdash \langle P,0 \rangle & \arros{loadA} & \langle (loadB \div unloadB);leave, unloadA \rangle\\
& \arros{loadB} & \langle leave, unloadB;unloadA \rangle\\
& \arros{0} & \langle \abort, unloadB;unloadA \rangle
\end{eqnarray*}
thus the process has aborted.
If we put the whole process inside a saga then compensation is actually executed and the saga succeeds:\\
\begin{eqnarray*}
\Gamma \vdash \langle \saga{P},0 \rangle & \arros{loadA}\arros{loadB} & \langle \sagarun{leave}{unloadB;unloadA},0 \rangle\\  
& \arros{0} & \langle \prot{unloadB;unloadA},0 \rangle\\
& \arros{unloadB} & \langle \prot{unloadA},0 \rangle\\
& \arros{unloadA} & \langle \commit,0 \rangle
\end{eqnarray*}

We consider here another possible computation, so to clarify one of
the most tricky cases of the dynamic semantics. Assume that
instead of activity $loadA$ with compensation $unloadA$ we have a
sequential composition of two activities, $loadA1$ with compensation
$unloadA1$ and $loadA2$ with compensation $unloadA2$.
Assume also that activity
$loadB$ aborts, while the other activities
succeed. 
\begin{eqnarray*}
\Gamma \vdash \langle P'',0 \rangle & \arros{loadA1} & \langle (\sagarun{loadA2}{unloadA1} | loadB \div unloadB);leave,0 \rangle\\
& \arros{\dagger} & \langle \killed{unloadA1}, 0 \rangle\\
& \arros{unloadA1} & \langle \abort, 0 \rangle
\end{eqnarray*}
Here $loadB$ aborts when the parallel saga is still running, thus
abortion is postponed. After compensation has been performed, the
abortion is raised again.
\end{example} 

\begin{example}\label{ex:subtransdynamic}
We show now the behavior of the saga in Example~\ref{ex:subtransstatic} when
executed under the dynamic semantics.
\begin{eqnarray*}
\Gamma \vdash \langle P',0 \rangle\! & \arros{loadA1} & \langle \sagarun{loadA2 \div unloadA2}{unloadA1} | loadB1 \div unloadB1;loadB2 \div unloadB2, 0 \rangle
\end{eqnarray*}
Here $unloadA1$ is not enabled, and it becomes enabled only after $loadB1$ is observed and $loadB2$ is executed, triggering the abort:
\begin{eqnarray*}
\Gamma \vdash \langle P',0 \rangle\!\!\!\!\! & \arros{loadA1} & \langle \sagarun{loadA2 \div unloadA2}{unloadA1} | loadB1 \div unloadB1;loadB2 \div unloadB2, 0 \rangle\\
& \arros{loadB1} & \langle \sagarun{loadA2 \div unloadA2}{unloadA1} | loadB2 \div unloadB2, unloadB1 \rangle\\
& \arros{\dagger} & \langle \killed{unloadA1}, unloadB1 \rangle\\
& \arros{unloadA1} & \langle \abort, unloadB1 \rangle
\end{eqnarray*}

\end{example}

\section{Static vs Dynamic SAGAs}\label{sec:comp}
In this section we compare the static and dynamic semantics of nested
SAGAs. In particular, we show that each computation obtained from the
dynamic semantics is compatible with a big-step of the static
semantics. We show also that the static semantics allows for more
nondeterminism in the order of execution of activities, i.e. it allows
for some computations not valid according to the dynamic semantics.

Labels of a big-step correspond to sets of computations of small-steps
(assuming an interleaving interpretation for parallel composition).
We write $\Arro{\alpha} $ with $\alpha=a_1;\dots;a_n$ to denote
$\arros{a_1}\cdots\arros{a_n}$. We remove from $\alpha$ both $0$ and
$\dagger$. However, big-step labels may also include parallel
composition operators, thus to perform the comparison we have to
introduce the concept of linearization.  We consider the set of
linearizations $\lin(\alpha)$ of a big-step label $\alpha$, which is
defined by structural induction on $\alpha$:
\begin{eqnarray*}
\lin(A) & = & \{A\}\\ 
\lin(\alpha;\alpha') & = & \{\gamma;\gamma' |\ \gamma \in \lin(\alpha) \land \gamma' \in \lin(\alpha')\}\\
\lin(\alpha|\alpha') & = & \bigcup_{\gamma \in \lin(\alpha) \land \gamma' \in \lin(\alpha')}\gamma \interleave \gamma'
\end{eqnarray*}
where $\interleave$ is defined as follows:
\begin{eqnarray*}
0 \interleave \gamma & = & \{\gamma\}\\
\gamma \interleave 0 & = & \{\gamma\}\\
A;\gamma \interleave A';\gamma' & = & \{A;\gamma''|\ \gamma'' \in \gamma \interleave A';\gamma'\} \cup \{A';\gamma''|\ \gamma'' \in A;\gamma \interleave \gamma'\}
\end{eqnarray*}
In words, $\interleave$ computes the set of all possible interleavings of the sequences of actions in its two arguments.

Summarizing, a big-step label $\alpha$ corresponds to the set of
small-step computations with labels in $\lin(\alpha)$.

Next lemma discusses the properties of $\dagger$-labeled transitions.


\begin{lemma}\label{lemma:dagger}
If $\Gamma \vdash \langle P,\beta \rangle \Arro{\dagger} \langle P',
\beta' \rangle$ then $P'$ is a parallel composition of terms of the
form $\killed{P_i}$.
\end{lemma}

We can now prove our main theorems, relating the behavior of static
and dynamic semantics for SAGAs. We have two theorems, one for each
direction.

\begin{theorem}\label{th:compatible}
If $\Gamma \vdash \langle P,\beta \rangle \Arro{\gamma} \langle
\dontcare, \beta' \rangle$ with $\dontcare \in
\{\commit,\abort,\crash\}$ then there is a big-step $\Gamma \vdash
\langle P,\beta \rangle \arrob{\alpha} \langle \dontcare, \beta''
\rangle$ with $\gamma \in \lin(\alpha)$ and $\beta' \in
\lin(\beta'')$.
\end{theorem}
\begin{proof}
The proof is by structural induction on $P$. Actually for the
induction we need a stronger hypothesis, requiring also that: 
\begin{itemize}
\item if $\Gamma \vdash \langle P,\beta \rangle \Arro{\gamma} \langle
  P', \beta' \rangle$ and $\Gamma \vdash \langle \extr{P'},0 \rangle
  \Arro{\gamma'} \langle \commit,0 \rangle$, then there is a big-step
  $\Gamma \vdash \langle P,\beta \rangle \arrob{\alpha;\alpha'}
  \langle \fabort, \beta'' \rangle$ with $\gamma \in \lin(\alpha)$,
  $\gamma' \in \lin(\alpha')$ and $\beta' \in \lin(\beta'')$;
\item if $\Gamma \vdash \langle P,\beta \rangle \Arro{\gamma} \langle
  P', \beta' \rangle$ and $\Gamma \vdash \langle \extr{P'},0 \rangle
  \Arro{\gamma'} \langle \dontcare,0 \rangle$ with $\dontcare \in \{\abort,\crash\}$, then there is a big-step
  $\Gamma \vdash \langle P,\beta \rangle \arrob{\alpha;\alpha'}
  \langle \fabcrash, 0 \rangle$ with $\gamma \in \lin(\alpha)$ and
  $\gamma' \in \lin(\alpha')$;
\item if $\Gamma \vdash \langle P,\beta \rangle \Arro{\gamma} \langle
  P', \beta' \rangle$, then there is a big-step $\Gamma \vdash \langle
  P,\beta \rangle \arrob{\alpha} \langle \fcrash, 0 \rangle$ with
  $\gamma \in \lin(\alpha)$.
\end{itemize}
We have the following cases:
\begin{description}
\item[$P=0$:] the only non trivial computation is $\Gamma \vdash
  \langle 0,\beta \rangle \arros{0} \langle \commit, \beta
  \rangle$. The big-step $\Gamma \vdash \langle 0,\beta \rangle
  \arrob{0} \langle \commit, \beta \rangle$ derived from rule
  \rn{zero} satisfies the thesis.  As far as the empty computation is
  concerned the two big-steps $\Gamma \vdash \langle 0,\beta \rangle
  \arrob{0} \langle \fabort, \beta \rangle$ and $\Gamma \vdash \langle
  0,\beta \rangle \arrob{0} \langle \fcrash, 0 \rangle$, derived from
  rules \rn{forced-abt} and \rn{forced-fail} respectively, satisfy
  the thesis.
\item[$P=A \div B$:] we have a case analysis according to
  $\Gamma(A)$. If $\Gamma(A)=\commit$ then the only non trivial
  computation is $\Gamma \vdash \langle A \div B,\beta \rangle
  \arros{A} \langle \commit, B;\beta \rangle$, and we have a
  corresponding big-step, derived from rule \rn{s-act}. Similarly for
  the case $\Gamma(A)=\abort$, using rule \rn{f-act}.  The empty
  computations can be matched as before.
\item[$P=P_1;P_2$:] assume that there is a computation $\Gamma \vdash \langle
  P_1;P_2,\beta \rangle \Arro{\gamma} \langle \dontcare, \beta'
  \rangle$ with $\dontcare \in \{\commit,\abort,\crash\}$. We have to
  consider the three cases $\dontcare=\commit$, $\dontcare=\abort$ and
  $\dontcare=\crash$.
\begin{description}
\item[$\dontcare=\commit$:] let us consider the first part of the
  computation. The only possibility is to have the first zero or more
  steps derived using as last rule \rn{step-d} followed by one step
  derived using as last rule \rn{s-step-d}. By concatenating the
  premises we have a computation $\Gamma \vdash \langle P_1,\beta
  \rangle \Arro{\gamma'} \langle \commit, \beta'' \rangle$. By
  inductive hypothesis we have a big-step $\Gamma \vdash \langle
  P_1,\beta \rangle \arrob{\alpha'} \langle \commit, \delta'' \rangle$
  with $\gamma' \in \lin(\alpha')$ and $\beta'' \in
  \lin(\delta'')$. Also, using the last part of the computation we have
  a big-step $\Gamma \vdash \langle P_2,\beta'' \rangle
  \arrob{\alpha''} \langle \commit, \delta' \rangle$ with $\gamma' \in
  \lin(\alpha'')$, $\beta' \in \lin(\delta')$ and
  $\gamma=\gamma';\gamma''$. The thesis follows by rule \rn{s-step}.
  It is not possible to have the first zero or more steps derived
  using as last rule \rn{step-d} followed by one step derived using as
  last rule \rn{k-step-d} since this computation can not succeed (see
  Lemma~\ref{lemma:dagger}).
\item[$\dontcare=\abort$:] we have a few possibilities here, according
  to which is the first rule applied different from \rn{step-d}. If it
  is rule \rn{s-step-d} then the thesis follows by inductive
  hypothesis applying rule \rn{s-step}. If it is rule \rn{a-step-d}
  then the thesis follows by inductive hypothesis applying rule
  \rn{a-step}. If it is rule \rn{k-step-d} then the thesis follows by
  inductive hypothesis, again applying rule \rn{a-step}.
\item[$\dontcare=\crash$:] similar to the one above.
\end{description}
  As far as computations leading to processes are concerned, a similar
  reasoning can be done. The thesis follows from rule \rn{a-step}
  if the computation is only composed by steps from rule
  \rn{step-d}.  If the computation includes also a step from rule
  \rn{s-step-d} then the thesis follows from rule \rn{s-step}. If the
  computation includes also a step from rule \rn{k-step-d} then the
  thesis follows from rule \rn{a-step}. 
\item[$P=P_1|P_2$:] assume that there is a computation $\Gamma \vdash \langle
  P_1|P_2,\beta \rangle \Arro{\gamma} \langle \dontcare, \beta'
  \rangle$ with $\dontcare \in \{\commit,\abort,\crash\}$. We have to
  consider the three cases $\dontcare=\commit$, $\dontcare=\abort$ and
  $\dontcare=\crash$.
\begin{description}
\item[$\dontcare=\commit$:] let us consider the first part of the
  computation. The only possibility is to have the first zero or more
  steps derived using as last rule \rn{par-d} followed by one step
  derived using as last rule \rn{s-par-d}. Assume for simplicity that
  \rn{s-par-d} eliminates the first component of the parallel
  composition. By concatenating the premises of those transitions that
  concern the first component we have a computation $\Gamma \vdash
  \langle P_1,0 \rangle \Arro{\gamma'} \langle \commit, \beta''
  \rangle$. By inductive hypothesis we have a big-step $\Gamma \vdash
  \langle P_1,0 \rangle \arrob{\alpha'} \langle \commit, \delta''
  \rangle$ with $\gamma' \in \lin(\alpha')$ and $\beta'' \in
  \lin(\delta'')$. Also, using the premises of the transitions
  involving the second component and the last part of the computation
  we have a big-step $\Gamma \vdash \langle P_2,\beta \rangle
  \arrob{\alpha''} \langle \commit, \delta''' \rangle$ with $\gamma''
  \in \lin(\alpha'')$ and $\beta''' \in \lin(\delta''')$. Also, $\gamma$ is an
  interleaving of $\gamma'$ and $\gamma''$ and $\beta'$ is obtained
  by prefixing $\beta$ with an interleaving of $\beta''$ and
  $\beta'''$. The thesis follows by rule \rn{s-par} (since $\commit \land
  \commit = \commit$).  It is not possible to have the first zero or
  more steps derived using as last rule \rn{step-d} followed by one
  step derived using as last rule \rn{k-step-d} since this computation
  can not succeed (see Lemma~\ref{lemma:dagger}).
\item[$\dontcare=\abort$:] we have a few possibilities here, according
  to which is the first rule applied different from \rn{par-d}. If it
  is \rn{s-par-d} then the thesis follows by inductive hypothesis
  applying rule \rn{s-par} (since $\commit \land \abort = \abort$). If it
  is \rn{a-par-d} then from the premises concerning $P$ we can derive
  an abortion for $P$. From $Q$ we can derive a computation leading to
  $Q'$ and an abortion for $\killed{\extr{Q'}}$, i.e. a commit for
  $\extr{Q'}$. Thus we can derive a big-step leading to $\fabort$ for
  $Q$. The thesis follows from rule \rn{s-par} (since $\abort \land
  \fabort= \abort$). The case of rule \rn{k-par-d} is similar, with
  the abortion coming after all compensations have been consumed
  (i.e., when rule \rn{a-par-fin-d} is triggered).
\item[$\dontcare=\crash$:] we have a few possibilities here, according
  to which is the first rule applied different from \rn{par-d}. If it
  is \rn{s-par-d} then the thesis follows by inductive hypothesis
  applying rule \rn{f-par} (since $\commit \land \crash = \crash$). If it
  is \rn{a-par-d} then from the premises concerning $P$ we can derive
  an abortion for $P$. From $Q$ we can derive a computation leading to
  $Q'$ and a failure for $\killed{\extr{Q'}}$, i.e. an abort or a
  failure for $\extr{Q'}$. Thus we can derive a big-step leading to
  $\fabcrash$ for $Q$. The thesis follows from rule \rn{f-par} (since
  $\abort \land \fabcrash= \crash$). The case of rule \rn{k-par-d} is
  similar. The case of rule \rn{f-par-d} follows from rule \rn{f-par},
  since $\crash \land \fcrash = \crash$.
\end{description}
  As far as computations leading to processes are concerned, a similar
  reasoning can be done. One has to distinguish interrupt from abort
  and interrupt from failure. In the first case the thesis follows by
  inductive hypothesis applying rule \rn{s-par} (with $\fabort \land
  \fabort= \fabort$) if both the compensations succeed. If at least
  one of the compensations fails then the other one is interrupted by a
  failure and the thesis follows by inductive hypothesis applying rule
  \rn{f-par} (with $\fabcrash \land \fcrash= \fabcrash$). If interruption is from a failure, then
  the thesis follows from rule \rn{f-par} (with $\fcrash \land \fcrash= \fcrash$).
\item[$P=\saga{P_1}$:] assume that there is a computation $\Gamma \vdash \langle
  \saga{P_1},\beta \rangle \Arro{\gamma} \langle \dontcare, \beta'
  \rangle$ with $\dontcare \in \{\commit,\abort,\crash\}$. We have to
  consider the three cases $\dontcare=\commit$, $\dontcare=\abort$ and
  $\dontcare=\crash$.
\begin{description}
\item[$\dontcare=\commit$:] we have two possibilities here. The first
  one is to have the first zero or more steps derived using as last
  rule \rn{saga-d} followed by one step derived using as last rule
  \rn{s-saga-d}. By concatenating the premises of those transitions we
  have a computation $\Gamma \vdash \langle P_1,0 \rangle
  \Arro{\gamma} \langle \commit, \beta' \rangle$. By inductive
  hypothesis we have a big-step $\Gamma \vdash \langle P_1,0 \rangle
  \arrob{\alpha} \langle \commit, \delta' \rangle$ with $\gamma \in
  \lin(\alpha)$ and $\beta' \in \lin(\delta')$. The thesis follows by
  rule \rn{sub-cmt}. The second case exploits rule \rn{saga-d} at the
  beginning (or possibly \rn{k-saga-d}), then one transition from rule
  \rn{a-saga-d}, then some transitions from rule \rn{prot-d} and
  finally one transition from rule \rn{s-prot-d}. Here by considering
  the premises of the first part of the computation we have a
  computation $\Gamma \vdash \langle P_1,0 \rangle \Arro{\gamma'}
  \langle \abort, \beta'' \rangle$. From the second part we get a
  computation $\Gamma \vdash \langle \beta'',0 \rangle \Arro{\gamma''}
  \langle \commit, 0 \rangle$. The thesis follows from the inductive
  hypothesis by applying rule \rn{sub-abt}.
\item[$\dontcare=\abort$:] there is no possibility for a saga
  computation to lead to an abort, thus this case can never happen.
\item[$\dontcare=\crash$:] we have two possibilities here. The first
  one is to have the first zero or more steps derived using as last
  rule \rn{saga-d} (and possibly some \rn{k-saga-d}) followed by one
  step derived using as last rule \rn{f-saga-d}. By concatenating the
  premises of those transitions we have a computation $\Gamma \vdash
  \langle P_1,0 \rangle \Arro{\gamma} \langle \crash, 0 \rangle$. By
  inductive hypothesis we have a big-step $\Gamma \vdash \langle P_1,0
  \rangle \arrob{\alpha} \langle \crash, 0 \rangle$ with $\gamma \in
  \lin(\alpha)$. The thesis follows by rule \rn{sub-fail-1}. The
  second case exploits rule \rn{saga-d} at the beginning, then one
  transition from rule \rn{a-saga-d}, then some transitions from rule
  \rn{prot-d} and finally one transition from rule \rn{a-prot-d}. Here
  by considering the premises of the first part of the computation we
  have a computation $\Gamma \vdash \langle P_1,0 \rangle
  \Arro{\gamma'} \langle \abort, \beta' \rangle$. From the second part
  we get a computation $\Gamma \vdash \langle \beta',0 \rangle
  \Arro{\gamma''} \langle \abort, \beta'' \rangle$. The thesis
  follows from the inductive hypothesis by applying rule
  \rn{sub-fail-2}.
\end{description}
  As far as computations leading to processes are concerned, a similar
  reasoning can be done. If the computation is only composed by
  applications of rules \rn{saga-d} or \rn{k-saga-d} then the thesis
  follows by inductive hypothesis from rule \rn{sub-fail-1} or from
  rule \rn{sub-forced-2}.  If there is an application of rule
  \rn{a-saga-d} then the thesis follows from rule \rn{sub-forced-1}.
\end{description}
\end{proof}

The following theorem considers the other direction.

\begin{theorem}\label{th:allows}
If $\Gamma \vdash \langle P,\beta \rangle \arrob{\alpha} \langle
\dontcare, \beta'' \rangle$ with $\dontcare \in
\{\commit,\abort,\crash\}$ then there are $\gamma \in \lin(\alpha)$ and
$\beta' \in \lin(\beta'')$ such that $\Gamma \vdash \langle P,\beta \rangle \Arro{\gamma} \langle
\dontcare, \beta' \rangle$ 
\end{theorem}
\begin{proof}
The proof requires a case analysis similar to the one of
Theorem~\ref{th:compatible}. We omit the details.
\end{proof}

\begin{example}
Consider the ship example and its two executions at the beginning
of Example~\ref{ex:shipstatic} (with static semantics) and
Example~\ref{ex:shipdynamic} (with dynamic semantics). It is easy to
see for instance that $loadA;loadB \in \lin(loadA|loadB)$ and
$unloadB;unloadA \in \lin(unloadA|unloadB)$. Thus the dynamic
computation is compatible with the static big-step, for a suitable
choice of the interleaving of parallel activities.
\end{example}

Notably, it is not possible to prove Theorem~\ref{th:allows} by
requiring a computation \emph{for each} possible linearization of
$\alpha$. In fact the chosen linearization depends on the runtime
execution of activities. For instance, in Example~\ref{ex:shipdynamic}
it is not possible to have as observation $loadA;loadB$ and as
compensation $unloadA;unloadB$. Also, abortion of nested sagas is
managed in a different way, as shown in the example below.

\begin{example}
Consider the saga in Example~\ref{ex:subtransstatic}. According to the
static semantics it has a big-step with label
$(loadA1;unloadA1)|loadB1$.  In particular, $loadA1;unloadA1;loadB1$
is a possible linearization of $(loadA1;unloadA1)|loadB1$. However, as shown in
Example~\ref{ex:subtransdynamic}, there is no dynamic computation with
this behavior. There is however a dynamic computation compatible with
the big-step, as shown in Example~\ref{ex:subtransdynamic},
considering the linearization $loadA1;loadB1;unloadA1$.
\end{example}

These are the main differences between
static and dynamic SAGAs: among all the computations compatible with
the static semantics, only some are acceptable for the dynamic
semantics, and whether a computation is acceptable or not depends on
the relation between the order of execution of activities and of their
compensations and on the interplay between nesting and parallel composition.

This also explains why we used a small-step semantics for dynamic
SAGAs (while classic semantics for static SAGAs is big-step): it is
difficult to catch the dependencies between the order of execution of
activities and the order of execution of compensations using big-step
semantics. For instance, a rule such as \rn{s-par} in
Table~\ref{table:staticsem} tailored for dynamic SAGAs should require
that the interleaving chosen for activities in $\alpha_1|\alpha_2$ is the
same as the one chosen for their compensations in $\beta_1|\beta_2$,
and one would need to track the correspondence between activities and
their compensations.

Summarizing, the theorems above provide the following insights:
\begin{itemize}
\item the static and the dynamic semantics are strongly related, in
  the sense that for each static big-step there is (at least) one
  dynamic computation compatible with it, and for each dynamic
  computation there is a compatible big-step;
\item the two semantics are not equivalent, since not all the
  computations compatible with a big-step are valid dynamic
  computations;
\item in the dynamic semantics the order of execution of compensations
  depends on the order of execution of basic activities, while this is
  not the case for the static semantics;
\item in the dynamic semantics compensations of subtransactions can be
  observed only after the abort itself has been observed.
\end{itemize}

\section{Conclusion}
We have presented two semantics for nested SAGAs, a static semantics
following the centralized interruption policy and a dynamic
semantics. The two semantics are a non trivial step forward w.r.t.\ the
non-nested case presented in the literature
(cfr. \cite{BBF+:SAGASCONCUR} and \cite{SEFM09}). Even the static
semantics is quite different from the (static) semantics of nested
SAGAs with distributed interruption presented in
\cite{BMM:SAGASPOPL}. The main difference relies in the different ways
compensations have to be managed in the two cases. Actually, we think
that the semantics with distributed interruption is realistic only in
asynchronous settings since, as said in \cite{BBF+:SAGASCONCUR}, it
``includes a ``guessing mechanism" that allows branches on the forward
flow to compensate even before an activity aborts''. This forbids for
instance an encoding into a lower-level framework such as the one in
\cite{SEFM09}. A similar behavior occurs also in our static semantics,
but only in the case of nested transactions.

As far as future work is concerned, many different proposals of
primitives and models for long-running transactions have been put
forward in the last years, yet a coherent picture of the field is
still far. Understanding the relationships between the different
formalisms is fundamental so to understand which of them are best
suited to be introduced in real languages. Even restricting our
attention to SAGAs, the interplay between dynamicity and the different
approaches presented in \cite{BBF+:SAGASCONCUR} has to be fully
analyzed. Also, we are currently working~\cite{BKLS:COORDINATED} on a
static semantics for SAGAs which is distributed but does not require
the guessing mechanism described above.

\section{Bibliography}

\bibliographystyle{eptcs} 
\bibliography{biblio}

\end{document}